\documentclass{amsart}
\usepackage[utf8]{inputenc}
\usepackage[english]{babel}

\usepackage[pass]{geometry}

\usepackage{graphicx}
\usepackage{siunitx}

\usepackage{comment}
\usepackage{url}

\usepackage[colorlinks=true,citecolor=blue,hypertexnames=false]{hyperref}
\usepackage{algorithm,algorithmic}

\usepackage{amsthm}

\newtheorem{defn}{Definition}
\newtheorem{proposition}[defn]{Proposition}
\newtheorem{lemma}[defn]{Lemma}
\newtheorem{theorem}[defn]{Theorem}

\newtheorem{corollary}[defn]{Corollary}

\newtheorem*{obs}{Empirical Observation}

\newcommand*{\nbset}[1]{\mathbb #1}

\newcommand*{\op}[1]{\mathcal #1}
\DeclareMathOperator{\ord}{ord}

\usepackage{xargs}
\newcommandx*{\suite}[3][2=n,3=n]{%
  \def\a{#3}                      %
  \def\nth{}                      %
  \ifx\a\nth                      %
  \def\res{\left(#1\right)_{#2\geq 0}}%
  \else                           %
  \def\res{\left(#1_{#3}\right)_{#2\geq 0}} %
  \fi%
  \res{}%
}

\DeclareMathOperator{\val}{val}
\def\lc{\operatorname{lc}}
\def\denom{\operatorname{den}}
\def\numer{\operatorname{numer}}

\DeclareMathOperator*{\K}{\mathcal K}
\newcommand*\cf[4]{\K_{#1}^{#2}\frac {#3}{#4}}

\usepackage{xcolor}

\def\paragraph#1{\smallskip\noindent{\bf #1.} }

\begin{document}

\title{Formulas for Continued Fractions\\An Automated Guess and Prove Approach}

\author{Sébastien Maulat}
\address{{ÉNS de Lyon}, LIP (U. Lyon, CNRS, ENS Lyon, UCBL), France}
\email{Sebastien.Maulat@ens-lyon.fr}

\author{Bruno Salvy}
\address{Inria, LIP (U. Lyon, CNRS, ENS Lyon, UCBL), France}
\email{Bruno.Salvy@inria.fr}

\begin{abstract}
We describe a simple method that produces automatically closed forms for the coefficients of continued fractions expansions of a large number of special functions. The function is specified by a non-linear differential equation and initial conditions. This is used to generate the first few coefficients and from there a conjectured formula. This formula is then proved automatically thanks to a linear recurrence satisfied by some remainder terms. Extensive experiments show that this simple approach and its straightforward generalization to difference and~$q$-difference equations capture a large part of the formulas in the literature on continued fractions.
\end{abstract}
\maketitle

\section{Introduction}
\label{sec:intro}

Continued fractions are well known for their approximation properties, their use in acceleration of convergence and analytic continuation, as well as their application in proofs of irrationality. Any formal power series can be converted into a \emph{corresponding} continued fraction (C-fraction)
\begin{equation}\label{eq:cf-generalform}
a_0+\cfrac{a_1(z)}{1+\cfrac{a_2(z)}{1+\cfrac{a_3(z)}{1+\dotsb}}}
\end{equation}
classically denoted~$a_0+\cf{m=1}{\infty}{a_m(z)}{1}$ or~$[a_0,a_1(z),a_2(z),\dotsb]$, where~$a_0$ is a constant and~$a_i(z)$ are nonconstant monomials for~$i>0$ that are called \emph{partial numerators}. In the frequent case when all the exponents are equal to~1, the C-fraction is called \emph{regular}. Truncating a continued fraction after its~$n$th term gives a rational function which is called its~$n$th \emph{convergent}. There is a one-to-one correspondence between power series and C-fractions. It is easily computed by a sequence of extractions of the constant coefficient, division by the variable and inversion. This conversion is available in the major computer algebra systems.

In several isolated cases, the coefficients~$a_m$ are known to possess a closed form, as in the following formula for~$\exp(z)$:
\[\exp(z)=1+\cfrac{z}{1-\cfrac{z/(2\cdot 1)}{1+\cfrac{z/(2\cdot 3)}{1-\cfrac{z/(2\cdot 3)}{1+\cfrac{z/(2\cdot5)}{1+\dotsb}}}}},\]
or more compactly
\begin{equation}\label{eq:cf-exp}
a_1=z,\;\; a_{2k}=-z/(2(2k-1)),\;\; a_{2k+1}=z/(2(2k+1)).
\end{equation}
Such formulas are the object of this work. A number of them are listed in the classical handbook by Abramowitz and Stegun~\cite{AbramowitzStegun1992}, or in its successor~\cite{OlverLozierBoisvertClark2010} and the most extensive list to date is the recent handbook by Cuyt \emph{et alii}~\cite{CuytPetersenVerdonkWaadelandJones2008}. Our aim is to derive many of these formulas automatically, starting from a description of the function to be expanded in continued fraction.

We concentrate on functions that are given as solutions of ordinary differential equations with initial conditions (or difference or~$q$-difference equations, see \S\ref{sec:expe}). Our approach can be summarized as follows. First, the differential equation and initial conditions are used to generate the first terms of the power series expansion of the function. This power series is then converted into a continued fraction. The coefficients of this continued fraction are then ``guessed'' by variants of rational interpolation. When this guessing phase is successful, a new power series is defined by this guessed continued fraction expansion. It remains to show that this power series satisfies the differential equation (the initial conditions being correct by construction). The key point in this proof is Theorem~\ref{prop:main}, stating that the (properly normalized) evaluations of the differential equation on the successive convergents to the continued fraction satisfy a linear recurrence, that can be computed. In all cases, after an operation we call ``reduction of order'', this recurrence exhibits a growth in the valuations that is sufficient to conclude the proof. A surprisingly large proportion of known explicit continued fractions are thus obtained completely automatically.

Classically, a very effective method due to Gauss derives formulas for continued fractions starting from the contiguity of hypergeometric series. 
Specialization of the parameters then leads to formulas for elementary or special functions~\cite[\S6.1]{JonesThron1980}. This leads to explicit continued fraction expansions by recognizing the function to be expanded as a special case of a quotient of contiguous hypergeometric functions and then relying on a small table of such explicit formulas.
These quotients satisfy Riccati equations, so that they are covered by our approach, which is not limited to them (see \S\ref{sec:expe}) and proves more suited to the targetted application to the \emph{Dynamic Dictionary of Mathematical Functions}~\cite{BenoitChyzakDarrasseGerholdMezzarobbaSalvy2010}. This is an online encyclopedia of special functions, where the formulas are all generated by computer algebra algorithms from differential equations, in many cases along with a human-readable proof. In this context, it makes sense to avoid any table lookup and generate formulas and proofs for continued fractions directly from the differential equation.

The work closer to ours is the investigation by Chudnovsky and Chudnovsky~\cite{ChudnovskyChudnovsky1991}. They used computer algebra in the study of formulas for continued fractions. Their aim was to classify all functions possessing continued fractions with explicit formulas of various types and relating them to Painlevé transcendents. In contrast, we focus on \emph{one} function that is given as input, and heuristically produce a rational continued fraction expansion when possible.

This article is structured as follows. Section~\ref{sec:example} gives an overview of our method on the example of the tangent function. Next, Section~\ref{sec:reduction} presents a heuristic of independent interest that reduces the order of a recurrence given initial conditions. This plays a crucial role in the proving phase of our method. Section~\ref{sec:guessing} is a brief account of what guessing means in this context, while Section~\ref{sec:proving} is the heart of this work and shows how proofs are achieved automatically. Finally, Section~\ref{sec:expe} presents experiments with this approach.

\section{Detailed Example: tan}\label{sec:example}

\noindent The tangent function can be defined  
by the Riccati equation
\begin{equation}
  \label{eq:de_tan}
  y'=1+y^2,\quad y(0)=0.
\end{equation}
The first 15 coefficients of the unique power series solution are easily computed from the differential equation~(see Proposition~\ref{lemma:val-increase} below for existence and uniqueness). A conversion into a continued fraction gives the coefficients 
\[[0,z,-z^2/3,-z^2/15,-z^2/35,-z^2/63,-z^2/99,-z^2/143].\]
The general formula can be deduced from these first terms by rational interpolation, which leads automatically to the (so far conjectural) formula 
\begin{equation}\label{eq:cfrac-tan}
a_1(z):=z; \quad a_n(z) := {-z^2}/({(2n-3)(2n-1)}), n>1.
\end{equation}

Next, we turn to the automatic proof of this formula. The strategy is to prove that the sequence of rational functions defined by truncating~\eqref{eq:cfrac-tan} after the~$n$th term for~$n=1,2,\dots$ converges to the formal power series solution to the differential equation~\eqref{eq:de_tan}. More precisely, let~$f_n$ be defined by
\[f_n=\frac{P_n}{Q_n}:=[0,z,-z^2/3,\dots,-z^2/((2n-3)(2n-1))],\]
where the rightmost term denotes the finite continued fraction. Then the proof will be completed by showing that~\\$\val (\tan -f_n)\rightarrow\infty$ as~$n\rightarrow\infty$, where~$\val$ denotes the valuation of a power series:
\[ \val\big(\sum_{i\geq 0}c_iz^i\big) := \min \{i\geq 0\mid c_i\neq 0\}, \]
with the convention~$\val(0)=\infty$.
Proposition~\ref{lemma:val-increase} below shows that it is sufficient to prove that~$\val(\mathcal D(f_n))\to \infty$, where\\$\mathcal{D}(f_n):=f_n'-1-f_n^2$.

It is classical that the numerator and denominator of the convergents of a continued fraction are related to the coefficients~$a_n$ through a linear recurrence:
\begin{equation}
  \label{eq:recPQ}
\begin{aligned}
    (P_{-1},P_0) &= (1,0),&  P_n&=P_{n-1}+a_nP_{n-2},& n&\geq 1,\\
    (Q_{-1},Q_0) &= (0,1),&  Q_n&=Q_{n-1}+a_nQ_{n-2},& n&\geq 1.
\end{aligned}
\end{equation}

In view of~\eqref{eq:cfrac-tan}, it follows that for all~$n\ge0$, $Q_n(0)=1$. Thus, the valuation of~$\mathcal{D}(f_n)$ is that of its numerator
\begin{equation}\label{eq:defH}
H_n:=P_n'Q_n-Q_n^2-P_n^2-P_nQ'_n.
\end{equation}
Using~\eqref{eq:recPQ} to rewrite~$P_{n+k}$ and~$Q_{n+k}$ in terms of~$P_n,P_{n+1}$, $Q_n,Q_{n+1}$, it follows that any shift~$H_{n+k}$ ($k\in\{0,1,2,\dotsc\}$) can be rewritten as a linear combination of
\[P_{n+i}'Q_{n+j}, P_{n+i}Q_{n+j}', P_{n+i}P_{n+j}, Q_{n+i}Q_{n+j},\] for~$i$ and~$j$ in~$\{0,1\}$. There are finitely many such terms, which implies that a linear dependency between~$H_n,H_{n+1},\dots$ (ie, a linear recurrence for~$H_n$) can be computed directly from~\eqref{eq:recPQ} by linear algebra. This computation produces a linear recurrence of order~4:
\begin{multline}\label{eq:bigrec}
( 2n+7 ) {z}^{8}H_n 
-{z}^{4} ( 2n+7 )  ( 2n+3 ) ^{2} ( 2n+1 ) ^{2}H_{n+1}\\
+2{z}^{2} ( 2n+5 )  ( 2n+3 ) ^{2} ( 2n+1 ) ^{2} ( 4{n}^{2}-{z}^{2}+20n+21 ) H_{n+2}\\
 - ( 2n+5 ) ^{2} ( 2n+1 ) ^{2} ( 2n+7 ) ^{2} ( 2n+3 ) ^{3}H_{n+3} \\
 + ( 2n+5 ) ^{2} ( 2n+1 ) ^{2} ( 2n+7 ) ^{2} ( 2n+3 ) ^{3}H_{n+4}=0.
\end{multline}
This recurrence is satisfied by all sequences defined by~\eqref{eq:defH}, with~$P_n$ and~$Q_n$ arbitrary solutions of~\eqref{eq:recPQ}. Using the actual sequences~$P_n$ and~$Q_n$ provided by the continued fraction gives the first values of~$H_n$:
\[-1,-z^2,-\frac{z^4}{9},-\frac{z^6}{225},-\frac{z^8}{11025},-\frac{z^{10}}{893025}.\]
From there, automatic guessing again suggests the following simpler recurrence for~$H_n$:
\begin{equation}\label{eq:small_rec_H}
  (2n+1)^2H_{n+1}-z^2H_n=0.
\end{equation}
And again, this recurrence admits of an automatic proof: the right Euclidean division of the fourth order recurrence operator from~\eqref{eq:bigrec} by this first order one has a remainder equal to~0.
 This shows that the solution of~\eqref{eq:small_rec_H} with the initial conditions given above coincides with the solution of~\eqref{eq:bigrec} with the same initial conditions, and thus the numerator of~$\mathcal{D}(f_n)$ satisfies~\eqref{eq:small_rec_H}. On this last recurrence, the increase of the valuation with~$n$ is clear and this concludes the proof that~$f_n$ converges to~$\tan$ and thus that the power series defined by the continued fraction~\eqref{eq:cfrac-tan} is that of~$\tan$.

In summary, starting with the differential equation~\eqref{eq:de_tan}, this method produces and proves automatically the general term of the famous continued fraction
\[\tan z = \cfrac{z}{1-\cfrac{z^2/3}{1-\cfrac{\ddots}{1-\cfrac{z^2/((2n-3)(2n-1))}{\ddots}}}}\]
that was the basis of Lambert's proof that~$\pi$ is irrational.

\section{Reduction of order by Guess and Prove}\label{sec:reduction}
The transformation of the large recurrence~\eqref{eq:bigrec} into the shorter one~\eqref{eq:small_rec_H} makes it possible to prove automatically that the valuations~$\val H_n(z)$ increase with~$n$. This transformation turns out to play a role in most of the examples dealt with in our experiments. It is actually of more general interest: the closure properties enjoyed by the class of D-finite series or P-recursive sequences give rise to operators satisfied by products or sums of zeroes of such operators~\cite{Stanley1999,SalvyZimmermann1994}. These operators annihilate all possible cases and are potentially of large size, while operators of smaller order may exist for the specific solution of interest. Such an operator may be a right factor of the large one and could be searched for by factoring, but this is made difficult by the potentially infinite number of distinct factorizations~\cite{Tsarev1996}.

\paragraph{Sequences}
Let~$\op A$ be a recurrence operator with polynomial coefficients in $n$, of order denoted by~$\ord \op A$ and leading coefficient~$\lc(\op A)(n)$.
A sequence~$\left(u_n\right)_{n\geq 0}$, abbreviated $\left(u_n\right)$, is said to be \emph{defined} by the operator~$\op A$ and the initial conditions~$\mathcal{K}=\left(u_i\right)_{i\in\mathcal{I}}$, when the value~$u_n$ is given by~$\mathcal{K}$ for~$n\in\mathcal{I}$ and by the recurrence operator evaluated at~$n-\ord \op A$ otherwise. Note that the set~$\mathcal{I}$ must contain
\begin{equation*}
  \{0,\dots,\ord \op A-1\}\cup \{i\in\mathbb{N}\mid \lc(\op A)(i-\ord \op A)=0\}.
\end{equation*}
  
\paragraph{Algorithm}
We now detail an efficient heuristic approach finding such right factors, whose complexity is controlled with respect to the size of the large operator.
The idea is to exploit the \emph{initial values} of the sequence by a ``guess and prove'' approach.
This is described in Algorithm~\ref{reducerecorder}. This algorithm takes as input a linear recurrence operator~$\op A$ and initial values, as well as an upper bound~$N$ on the number of coefficients used to find a right factor. It is described here in the case of recurrence operators; similar variants apply to differential or~$q$-difference cases. %

The search for a smaller order operator is performed in two main steps, ``guessing'' and ``proving''. 
First, the input recurrence of order~$M$ and its initial values are used to compute the first~$N$ terms of the sequence.
Next, these~$N$ terms are used to ``guess'' a linear recurrence. This is done by linear algebra: we search successively for the existence of a linear recurrence operator~$\op G$ of order~$1,2,\dots,M$ with polynomial coefficients of degrees such that the sum of the numbers of undetermined coefficients of the recurrence is smaller than~$N$. The structure of this linear algebra problem is exploited by computing matrix rational interpolants~\cite{BeckermannLabahn1997} (in the differential case, Hermite-Padé approximants are used~\cite{BeckermannLabahn1994}).

When~$N$ is sufficiently large, this linear algebra phase is always successful, since it can reconstruct~$\op{A}$. The next step is to prove that the recurrence~$\op G$ obtained from the first~$N$ terms of~$\left(u_n\right)$ defines the same sequence for all~$n$. The operator~$\op G$ is not necessarily a right factor of~$\op A$, but could be merely a left multiple of such a right factor, the factor itself being too large to be found with~$N$ terms only. This is related to the typical shape of the order-degree curve~\cite{ChenJaroschekKauersSinger2013}. Thus the algorithm next computes the greatest common right divisor of~$\op G$ and~$\op A$ and its {numerator}~$\op R$, 
obtained by left-multiplication with the least common multiple of the  denominators of the coefficients. 

At this stage, the algorithm has produced a right factor~$\op R$ of~$\op A$. It is then associated initial conditions~$\left(u_i\right)_{i\in\mathcal{J}}$,
with which $\op R$ defines a sequence~$\suite v[n][n]{}$. We now prove that if 
$v_n=u_n$ for~$n\in\mathcal{J}+\{1,\dots,\ord \op A-\ord\op R\}$, then $v_n=u_n$ for all $n$. The induction on $n$ is based on the following.
\begin{lemma} If $u_n=v_n$ for $n\le i+\ord A-1$ and
$i+j\notin\mathcal J$ for all $j\in\{\ord \op R,\dots,\ord \op A\}$ then $u_{i+\ord A}=v_{i+\ord A}$.
\end{lemma}
\begin{proof}
The sequences $\{S^j\op R\cdot v_n\}_{\ord \op R\leq j\leq \ord \op A}$ all cancel at $n=i$.
The application of $\op A$ at $\suite v$ is a linear combination of them with coefficients that are finite, as shown in the next lemma, so that $\op A\cdot v_n$ is 0 at $n=i$.
\end{proof}
\begin{lemma}
  Let~$\op A$ and~$\op C$ be recurrence operators with polynomial coefficients, satisfying~$\op A=\op B \op C$ where~$\op B$ has rational coefficients. Then the denominator~$\denom\op B$ satisfies:
  \[ \denom(\op B)^{-1}(0)
  \subseteq
  \lc(\op C)^{-1}(0)
  +\{0,-1,\dots,-\ord \op B\} \]
  where addition denotes the sumset.
\end{lemma}
\begin{proof}This is seen by following the steps of a right Euclidean division.
\end{proof}

\newfloat{algorithm}{tb}{lop}
\floatname{algorithm}{Algorithm}
\renewcommand{\algorithmicrequire}{\textbf{Input:}}
\renewcommand{\algorithmicensure}{\textbf{Output:}}
\newcommand{\mycmdstyle}{\textit}
\def\guessrec{\operatorname{guessrec}}
\def\firstterms{\operatorname{firstterms}}
\def\gcrd{\operatorname{gcd_{right}}}
\def\quotient{\operatorname{quo_{right}}}
\def\NULL{\operatorname{FAIL}}

\begin{algorithm}
\centering
\begin{minipage}{\linewidth}%
\flushleft
\begin{algorithmic}
  \REQUIRE$(\op A,\left( u_n \right)_{n\in \mathcal{I}})$ defining~$\suite u$, and~$N>0$.
  \ENSURE$(\op R,\left( u_n \right)_{n\in {\mathcal{J}}})$ defining~$\suite u$ {\it s.t.}~$\ord\op R\leq \ord\op A$.
  \STATE{$\mathcal{U} \gets (u_n)_{n=0,\dots,N-1}$, computed using~$\op A$ and~$\left( u_n \right)_{n\in {\mathcal{I}}}$}
  \STATE{$\op G\gets \guessrec(\mathcal{U})$}
  \IF{$\op G\neq \NULL$}%
  \STATE{$\op R\gets \numer\left(\gcrd(\op A, \op G)\right)$} 
  \STATE{${\mathcal{J}}\gets \mathcal{I} \cup \Big( \lc(\op R)^{-1}(0) -\ord \op R \Big) \cap \mathbb N$}
  \STATE{$\mathcal{V} \gets (v_n)_{n\in  \mathcal{J}+\{1,\dots,\ord \op A-\ord\op R\}}$, using~$(\op R,\left( u_n \right)_{n\in {\mathcal{J}}})$}
  \STATE{$\mathcal{U}' \gets (u_n)_{n\in  \mathcal{J}+\{1,\dots,\ord \op A-\ord\op R\}}$, using~$(\op A,\left( u_n \right)_{n\in {\mathcal{I}}})$}
  \algorithmicif~$\mathcal{U}'=\mathcal{V}$ \algorithmicthen~ \algorithmicreturn~ $(\op R,\left(u_n\right)_{n\in \mathcal{J}})$~\algorithmicendif
  \ENDIF
  \RETURN {$(\op A, \left(u_n\right)_{n\in \mathcal{I}})$}
\end{algorithmic}
\end{minipage}
\caption{Reduction of Order}
\label{reducerecorder}
\end{algorithm}

In practice, this algorithm is run for increasing values of~$N=4,8,16,\dots$ and stopped when either a factor is found or~$N$ is larger than the number of coefficients of~$\op A$. Note however that if~$\op A$ is not irreducible, then increasing~$N$ further is bound to find a nontrivial right factor.

\section{Guessing Continued Fractions}\label{sec:guessing}

The first step of our approach to continued fractions is the automatic discovery of formulas for the partial numerators~$a_k$.
This section is very short since this part of the computation is straightforward.

Starting from the differential equation, a first method would be to produce the first terms of the series expansion of the function, convert them into the first terms of the continued fraction and then use the method of the previous section to look for a linear recurrence of size bounded by the number of terms that have been computed. It turns out that in most of the known examples, 
explicit formulas are of rational form (see Section~\ref{sec:expe}).
We therefore concentrate on rational coefficients, or on interlacing of rational coefficients as in the example of the exponential function. This means that the ``guessing'' stage of our approach relies simply on rational interpolation, problem for which efficient algorithms are known through its relation to the extended Euclidean algorithm~\cite[\S5.7]{GathenGerhard2003}. Moreover, the degrees of the numerator and denominator are generally low, so that a few terms of the expansion are sufficient for the computation.

\section{Proving Continued Fractions\\for solutions of ordinary\\differential equations}\label{sec:proving}
\newcommand{\numers}{\suite {P_n(z)}[n][]}
\newcommand{\denoms}{\suite {Q_n(z)}[n][]}
\newcommand{\Hs}{\suite H[k][k]}
The proving phase is the heart of our work. It is also where most of the computational work takes place.
We consider first-order non-linear differential equations with rational coefficients, ie, $y'=p(y)$, with~$p\in \nbset Q(z)[Y]$ of degree~$d$. 
In particular, the case~$d\le 2$ corresponds to the classical Riccati equations that are ubiquitous in the study of continued fractions, due to their 
stability under linear fractional transformations of the unknown function~\cite[10.7]{BakerGraves-Morris1996}. Explicit solutions for restricted classes of equations have been provided by Euler and Lagrange and more recently by Khovanskii~\cite{Khovanskii1963}. 

Our procedure goes in the reverse direction. The continued fraction with explicit rational coefficients that was found in the previous stage defines a power series. The aim is to show that it is a solution of the differential equation. 

\subsection{Valuations}
The following proposition reduces the proof to that of the ultimate increase of a sequence of integers.
\begin{proposition}\label{lemma:val-increase}
Let~$F\in\mathbb{K}[[X,Y]]$ be a formal power series with coefficients in a field~$\mathbb{K}$ and let~$(f_n(X))$ be a sequence of power series in~$\mathbb{K}[[X]]$. Then the differential equation~$Y'=F(X,Y)$ with initial condition~$Y(0)=0$ admits a unique power series solution~$S(X)$. Moreover, the sequence~$(f_n(X))$ converges to~$S(X)$ (ie~$\val (f_n-S)\rightarrow\infty$) if and only if~$f_n(0)=0$ for~$n$ sufficiently large and~$\val(f_n'(X)-F(X,f_n(X)))\rightarrow\infty$.
\end{proposition}
Note that equations with an initial condition~$Y(0)=a\neq0$ can often be brought to this setting by changing the unknown function into~$a+Y$.\\
\begin{proof}Recall that the algebra of power series is a metric space for the distance induced by the valuation: if~$f$ and~$g$ are two power series, then~$d(f,g)=2^{-\val(f-g)}$, where~$\val$ denotes the valuation (this distance does not derive from a norm). It is a simple consequence of the definition that Cauchy sequences for this distance converge in~$\mathbb{K}[[X]]$.

The first part of the proposition is a variant of Cauchy's theorem, whose proof is straightforward thanks to 
Taylor expansions. In detail, the solutions of~$Y'=F(X,Y)$ with initial condition~$Y(0)=0$ are the fixed points of the operator~$\mathcal{G}:Y\mapsto\int{F(X,Y)}$;
this operator is a contraction:
\begin{multline*}
\val(\mathcal{G}(Y_1)-\mathcal{G}(Y_2))=\val\left(\int{F(X,Y_1)-F(X,Y_2)}\right)\\
=\val\left(\int{\frac{\partial F}{\partial Y}(X,Y_2)(Y_1-Y_2)+O((Y_1-Y_2)^2)}\right)
>\val(Y_1-Y_2);
\end{multline*}
this shows both the existence of a solution (start from~$Y=0$, iterate~$\mathcal G$ and use completeness) and its uniqueness. 

Next, if~$\val(f_n'-F(X,f_n))=K$, while~$f_n(0)=0=S(0)$, then
\begin{align*}
S-f_n&=S(0)-f_n(0)+\int{(F(X,S)-F(X,f_n))}+O(x^{K+1})\\
&=(\mathcal{G}(S)-\mathcal{G}(f_n))+O(x^{K+1}).
\end{align*}
The previous inequality with~$Y_1=S$ and~$Y_2=f_n$ shows that the valuation of the first term on the right-hand side is larger than that of the left-hand size and thus~$\val(S-f_n)\geq K+1$, which shows that~$f_n\rightarrow S$. The converse implication follows from the continuity of the map~$Y\mapsto Y'-F(X,Y)$.
\end{proof}
This proposition extends to more general equations of the type~$P(z,y,y',\ldots,y^{(n)})=0$, with natural assumptions on the initial and separant of the equation. 

\subsection{P-recursivity and Convergents}
Recall that a sequence is called P-recursive when it satisfies a linear recurrence with coefficients that are polynomial in the index. P-recursive sequences are closed under sum and product and algorithms computing the corresponding recurrences are known and implemented~\cite{Stanley1999,SalvyZimmermann1994}. 

The key to our approach is the following.

\begin{theorem}\label{prop:main}
Let~$(P_k(z))$ and~$(Q_k(z))$ be P-recursive sequences of rational functions in~$z$ and let~$F\in \mathbb{K}(z)[Y]$ be a polynomial of degree~$d>0$ in~$Y$. 
Then the sequence
  \[H_k:=Q_k^{\max(2,d)}\left(\left(\frac{P_k}{Q_k}\right)'-F\!\left(z,\frac{P_k}{Q_k}\right)\right)\]
satisfies a linear recurrence with coefficients in~$\mathbb{K}[z,k]$. 
\end{theorem}
This theorem is used when~$\left(P_k\right)$ and~$\left(Q_k\right)$ are the sequences of numerators and denominators of the continued fraction supposed to converge to a solution of~$y'=F(z,y)$. Its proof constructs a recurrence for~$H_k$ from which the increase of valuation will be obtained using the reduction of order of Section~\ref{sec:reduction}. This is sufficient thanks to Proposition~\ref{lemma:val-increase} and the observation that~$\val Q_k=0$, which will follow by induction from Eq.~\eqref{eq:recPQ} and the fact that~$\val a_k>0$ in applications to C-fractions.

Again, similar results can be stated for higher order differential equations, but they proved unnecessary for the continued fractions dealt with in our experiments.

\begin{proof}
Let~$M$ be the order of the recurrence satisfied by~$\left(P_k\right)$.
Using this recurrence, all~$P_{k+i}$, $i\in\mathbb{N}$ can be rewritten as linear combinations of~$P_{k+j}$ for~$j=0,\dots,M-1$,
with coefficients in~$\mathbb{Q}(k,z)$, while the polynomials~$P'_{k+i}$ rewrite as linear combinations of those same polynomials complemented by 
$P'_{k+j}$ for~$j=0,\dots,M-1$. The same argument applies to the sequence~$\left(Q_k\right)$ and we denote by~$M'$ the order of the recurrence it satisfies.

The choice of the exponent of~$Q_k$ makes~$H_k$ a 
polynomial of degree~$d$ in~$P_k$, $Q_k$, $P_k'$ and~$Q_k'$. Thus all the~$H_{k+i}$ for~$i\in\mathbb{N}$ can be rewritten as linear combinations of monomials of degree~$d$ in~$P_{k+i},Q_{k+j}$, $i=0,\dots,M-1$, $j=0,\dots,M'-1$ and their derivatives. These monomials are in finite number~$N$, whence a linear dependency between~$H_k,\dots,H_{k+N}$ (ie, a linear recurrence of order at most~$N$ satisfied by~$\left(H_k\right)$). It can be found by linear algebra.
\end{proof}

\newcommand{\gens}{\mycmdstyle{G}} 
\begin{algorithm}
  \caption{Recurrence for~$\suite H[k][k]$}
  \label{recH}
\begin{algorithmic}
  \REQUIRE linear recurrences~$\op L_P$ and~$\op L_Q$ of order bounded by~$M$ for~$\suite P[k][k]$ and~$\suite Q[k][k]$
  \ENSURE a linear recurrence~$\op L_H$ for~$\suite H[k][k]$
  \STATE$T_0(k)\gets H_k$
  \FOR{$i=1,2,3,\dots$}
  \STATE$T_i(k) \gets T_{i-1}(k+1)$ with~$P_{k+M}$, $P'_{k+M}$, $Q_{k+M}$, $Q'_{k+M}$ rewritten in terms of values of these sequences with smaller indices, using~$\op L_P$, $\op L_Q$ and their derivatives.
  \IF {the linear equation~$\sum_{j=0}^{i-1}{c_jT_j(k)}+T_i(k)$ in the unknowns~$c_0,\dots,c_{i-1}$ has a solution}
  \RETURN$H_{k+i}+c_{i-1}H_{k+i-1}+\dots+c_0H_k=0$
  \ENDIF
  \ENDFOR
\end{algorithmic}
\end{algorithm}
As exemplified by the computation in the example of~$\tan$ in Section~\ref{sec:example}, the bound on the order on specific examples may not be as large as suggested by this proof. Our implementation thus proceeds by increasing the order one by one and looking for a linear dependency until one is found. The algorithm is outlined in Algorithm~\ref{recH}. Its termination is granted by the theorem.

We state a simple generalization of this result that could be useful in applications to continued fractions: if the partial numerators in the continued fraction expansion~\eqref{eq:cf-generalform} are of the form $r(k)z^{p(k)}$ with~$r$ rational and~$p$ polynomial, then again, the polynomials~$H_k$ defined in the theorem satisfy a linear recurrence, this time with coefficients that are polynomials in~$k,z$ and a finite number of~$z^{k^i}$, with~$i\le\deg p$. The proof follows along the same lines.

\subsection{Riccati Equations}
The case when the polynomial~$F$ of Theorem~\ref{prop:main} has degree~2 gives rise to Riccati equations that are ubiquitous in the theory of continued fractions~\cite[10.7]{BakerGraves-Morris1996}. In this case, the computation of a recurrence of the form predicted by the theorem can be made explicit in full generality.
\begin{proposition}\label{prop:riccati}Let~$\cf{k=1}{\infty}{a_k(z)}{1}$ be a solution of the Riccati differential equation~$Y'=F(z,Y)$ where~$F$ is a polynomial in~$\mathbb{K}(z)[Y]$ of degree~2 in~$Y$, let~$(P_k)$ and~$(Q_k)$ be sequences obeying the linear recurrence~$u_{k+2}=u_{k+1}+a_{k+2}(z)u_k$, with~$a_k'(z)\neq0$. Let finally~$H_k$ be defined by
\[H_k=Q_k^2\left((P_k/Q_k)'-F(z,P_k/Q_k)\right).\]
Then the sequence~$(H_k)$ satisfies the following linear recurrence of order~4:
\begin{multline*}
\frac{1}{a'_{k+1}}H_{k+1}+\left(\frac{a_k}{a_k'}-\frac{a_{k+1}+1}{a'_{k+1}}\right)H_k\\
-\left(\frac{a_k(a_k+1)}{a'_k}+\frac{a_{k+1}(a_{k+1}+1)}{a'_{k+1}}\right)H_{k-1}\\
-\left(\frac{a_k+1}{a'_k}-\frac{a_{k+1}}{a'_{k+1}}\right)a_k^2H_{k-2}
+\frac{a_{k-1}^2a^2_k}{a'_k}H_{k-3}=0.
\end{multline*}
\end{proposition}
The shift of the indices (from~$H_{k+1}$ to~$H_{k-3}$) is only for readability.
A nice property of this recurrence is that its coefficients do not depend on the differential equation, but only on the sequence~$a_k$. This persists for higher degrees: a differential equation with a cubic right-hand side leads to a recurrence of order~6 that does not depend on the equation.
\begin{proof}The formula is obtained automatically by the method from the proof of Theorem~\ref{prop:main}, on a differential equation with symbolic coefficients. It could also be derived by hand. However, once it is given, it is a simple matter to produce a proof: inject the definition of~$H_k$ into the recurrence, rewrite all the~$P_k$'s and~$Q_k$'s using the recurrence they satisfy in terms of~$P_{k-3},Q_{k-3},P_{k-2},Q_{k-2}$ and collect terms to observe that the left-hand side becomes~0.
\end{proof}
As an example, setting~$a_k(z)=-z^2/((2k-1)(2k-3))$ recovers Eq.~\eqref{eq:bigrec} obtained for~$\tan$. 

\begin{corollary}\label{coro}
If the sequences~$(P_k)$ and~$(Q_k)$ satisfy a linear recurrence of the form~$u_{k+2}=b_{k+2}(z)u_{k+1}+a_{k+2}(z)u_k$, then the sequence~$(H_k)$ satisfies a fourth-order linear recurrence obtained by evaluating that of Prop.~\ref{prop:riccati}, replacing~$a_1$ by~$a_1/b_1$ and~$a_k$ by~$a_k/(b_kb_{k-1})$ for~$k\geq 2$.
\end{corollary}
\begin{proof}
This is a classical transformation of continued fractions. Setting~$\tilde{P}_k=P_{k}/(b_1\dotsm b_{k-1}b_k)$ and similarly for~$\tilde{Q}_k$ and injecting into the recurrence equation shows that both sequences~$(\tilde{P}_k)$ and~$(\tilde{Q}_k)$ satisfy
\[u_{k+2}=u_{k+1}+\frac{a_{k+2}}{b_{k+2}b_{k+1}}u_k.\]
Since~$\tilde{P}_k/\tilde{Q}_k=P_k/Q_k$, the proposition applies.
\end{proof}

\subsection{Nonregularity and Periodicities}\label{sec:periods}
As the example of the continued fraction for the exponential function in Eq.~\eqref{eq:cf-exp} shows, not all common closed forms for continued fractions are given by one rational function. However, most C-fractions formulas in the literature appear to be ``periodic'', in the sense that there exists a period~$\ell>0$, and~$\ell$ sequences~$(a^0_k),\dots,(a^{\ell-1}_k)$, that alternately define the partial numerators~$a_k$:~$a_k=a^{(k\mod \ell)}_k$. The case~$\ell=2$ encountered for~$\exp$ is the most common, but higher values also happen (e.g., $\ell=4$ for~$\psi''$, where~$\psi=\Gamma'/\Gamma$ is the logarithmic derivative of the Gamma function).

This is not a restriction in our approach, by the following.
\begin{lemma} 
  \label{holo_period}
  Given a period~$\ell>0$, a sequence~$\suite u[k][k]$ is P-recursive if and only if all its subsequences~$\suite {u}[k][\ell k+j]$ are P-recursive,  for~$j=0\ldots \ell-1$.
\end{lemma}
\begin{proof}This lemma is classical. We give a constructive proof for completeness. If the sections~$\suite{u}[k][\ell k+j]$ are P-recursive, then their generating series~$s_j(z)=\sum_{k\geq 0}{u_{\ell k+j}z^k}$ are D-finite, then so is~$s_0(z^\ell)+zs_1(z^\ell)+\dots+z^{\ell-1}s_{\ell-1}(z^\ell)$ and therefore its sequence of coefficients~$\suite u[k][k]$ is P-recursive. Conversely, if~$\suite u[k][k]$ is P-recursive, then its generating series~$s(z)$ is D-finite and so is its Hadamard product with $z^j/(1-z^\ell)$, then also its quotient by~$z^j$ evaluated at~$z^{1/\ell}$ and this is precisely the generating series of~$\suite u[k][\ell k+j]$.
\end{proof}

In cases like the exponential function, this lemma implies that the sequence of partial numerators~$\suite a[k][k]$ itself satisfies a linear recurrence. With the recurrences~\eqref{eq:recPQ}, this alone does not imply that~$(P_k)$ and~$(Q_k)$ are also P-recursive for in general, no such closure property exists. For instance, the sequence defined by~$u_n:=\prod_{k=1}^n{k!}$ satisfies a linear recurrence of order~1 with a coefficient ($k!$) that is P-recursive, but~$(u_n)$ itself is not P-recursive, as can be seen from its asymptotic growth that is too fast.
The crucial property in our application is that the sequences~$(a_k^i)$ are rational in~$k$. This allows for the following.
\begin{lemma}
\label{PQ_dfinite}
Let~$a^0_k, \ldots, a^{\ell-1}_k$ be rational functions of~$k$ and~$z$, $a_k$ be defined for~$k\geq 0$ by 
 ~$a_k := a^{(k\mod \ell)}_k$ and let the sequences~$(P_k)$ and~$(Q_k)$ be defined by the recurrence~\eqref{eq:recPQ}. Then~$(P_k)$ and~$(Q_k)$ are P-recursive sequences.
\end{lemma}
\begin{proof}
The proof is constructive. 
By the recurrences~\eqref{eq:recPQ} and the definition of~$a_k$, all~$P_{\ell k+j}$ for~$j=1,\dots,2\ell$ rewrite as linear combinations of~$P_{\ell k+1}$ and~$P_{\ell k}$ with coefficients that are rational functions of~$z$ and~$k$. Thus~$P_{\ell k}+j$, $P_{\ell(k+1)+j}$ and~$P_{\ell(k+2)+j}$ are linearly dependent for~$j=0,\dots,d-1$. 

The same reasoning applies to~$(P_{\ell k+j})_{k\geq 0}$ for any $j\in\{0,\dots,\ell -1\}$ and shows that it is a P-recursive sequence and thefore that so is~$(P_k)$ by the previous lemma. By construction, $(Q_k)$ satisfies the same recurrence as~$(P_k)$.
\end{proof}

\paragraph{Example}The special case~$\ell=2$ is important in applications. Starting from
$P_{2k}=P_{2k-1}+a_{2k}P_{2k-2}$ and its first two shifts, 
the linear combination~$P_{2k+2}+P_{2k+1}-a_{2k+1}P_{2k}$ gets rid of the terms with odd index, leaving:
\begin{equation}\label{eq:rec-even}
P_{2k+2}=(1+a_{2k+1}+a_{2k+2})P_{2k}-a_{2k}a_{2k+1}P_{2k-2}.
\end{equation}
A similar computation would give a recurrence between the terms with odd index.
\smallskip

The proof of Lemma~\ref{PQ_dfinite} leads to an algorithm in two steps: compute a recurrence for~$(P_k)$ and~$(Q_k)$ and then appeal to Lemma~\ref{holo_period}. A simpler and faster computation proceeds directly from a recurrence for~$(P_{\ell k})$, thanks to the following. 
\begin{proposition}Let~$F(X,Y)$ be a rational function, 
that is regular at~$X=Y=0$. Let~$a^i_k$, $i=1,\dots,\ell$ be rational functions in~$X$ and~$k$ with positive valuation in~$X$. Let~$(a_k)$, $(P_k)$ and~$(Q_k)$ be defined as in the previous lemma and~$(H_k)$ as in Theorem~\ref{prop:main}. If~$\val H_{k\ell}\rightarrow\infty$ as~$k\rightarrow\infty$, then the continued fraction~$\cf{m=1}{\infty}{a_k}{1}$ is the solution of~$Y'=F(X,Y)$ with~$Y(0)=0$.
\end{proposition}
Here again, other values for~$Y(0)$ are obtained by a change of unknown function.
\begin{proof}Since the denominator of~$F$ does not vanish at~$0$, $F$ admits an expansion in power series and thus by Proposition~\ref{lemma:val-increase}, the differential equation~$Y'=F(X,Y)$ with~$Y(0)=0$ possesses a formal power series solution~$S(X)$. 

The condition on the valuations of the sequences~$(a_k^i)$ makes the continued fraction well-defined, in the sense that the sequence of power series~$(P_k/Q_k)$ converges to a power series~$G(X)$. Thus if its subsequence~$f_k=P_{k\ell}/Q_{k\ell}$ converges to~$S(X)$, then we have~$G(X)=S(X)$.
Induction from Eq.~\eqref{eq:recPQ} shows that~$Q_k(0)=1$ for all~$k\ge0$ and that~$P_k(0)=0$ (by the positive valuation of~$a_1^1$). This gives~$f_k(0)=0$ and~$\val H_{k\ell}=\val(f_k'-F(X,f_k))$. Thus by Proposition~\ref{lemma:val-increase}, the sequence~$(f_k)$ converges to~$S$. 
\end{proof}

\paragraph{Example}The proof for the continued fraction for $\exp$ from the introduction goes as follows.
Starting from the recurrences for~$(P_k)$ when~$k$ is even and when~$k$ is odd and proceeding as for Eq.~\eqref{eq:rec-even} yields
\[P_{2k+2}=P_{2k}+\frac{z^2}{4(4k^2-1)}P_{2k-2},\]
which is also satisfied by~$Q_{2k}$ since this computation does not depend on the initial conditions.

Next, turn to the numerator of the evaluation of~$y'-y-1$ at~$y=P_{2k}/Q_{2k}$, namely
\[H_{2k}=P_{2k}'Q_{2k}-P_{2k}Q_{2k}'-P_{2k}Q_{2k}-Q_{2k}^2.\]
Using Proposition~\ref{prop:riccati}, or directly as in Section~\ref{sec:example}, leads to a recurrence of order~4, 
on which reduction of order yields
$H_{2k+2}=-{z^2}H_{2k}/({4(2k+1)^2})$,
which concludes the proof.

\section{Experiments}\label{sec:expe}

An overview of the whole approach is given in Figure~\ref{fig:method}. In practice, $N=20$ and~$L=2$ have proved sufficient in our experiments except for one case of period~4. For the computation of the first terms of the continued fraction, one can either compute a power series expansion first, e.g., by Newton iteration~\cite{BrentKung1978}, or use techniques for continued fraction expansions of solutions of Riccati equations~\cite{CooperCooperJones1991}.
\begin{algorithm}
\begin{algorithmic}
  \REQUIRE~$Y'=F(z,Y)$ with~$F\in\mathbb{K}(z)[Y]$;\\
 a bound~$N$ on the number of coefficients to guess from;\\
 a bound~$L$ on the period to be found.
 \ENSURE In case of success, an explicit expression for the continued fraction expansion of the solution such that~$Y(0)=0$.
 \vspace{.3em}
  \STATE Compute the first coefficients~$a_1,\dots,a_N$ of the continued fraction expansion of the power series solution of~$Y'=F(z,Y)$ with~$Y(0)=0$.
  \FOR{$\ell=1,2,3,\dots,L$}
  \STATE{Use rational interpolation to compute~$a_i^j$ interpolating the subsequences~$(a_{\ell i+j})_i$, $j=0,\dots,\ell-1$.}
  \IF {this has been successful}
    \STATE{compute a recurrence~$\mathcal{R}$ for~$H_{k\ell}$, with~$H_k$ defined in Theorem~\ref{prop:main}}
    \STATE{compute a new recurrence~$\mathcal{R'}$ from~$\mathcal{R}$ and the initial conditions for~$(H_k)$ using Algorithm~\ref{reducerecorder}}
    \IF {$\mathcal{R'}$ exhibits the increase of~$(\val H_{\ell k})$}
      \RETURN{the rational functions~$a_i^j$}
    \ENDIF
  \ENDIF
  \ENDFOR
  \RETURN{FAIL}
\end{algorithmic}
\caption{Discovery and proof of continued fractions}
\label{fig:method}
\end{algorithm}

Our main experimental result is the following.
\begin{obs}
All the 53 explicit C-fractions formulas of the compendium by Cuyt {\it et alii}~\cite{CuytPetersenVerdonkWaadelandJones2008}
can be guessed and proved by our approach and its variants below. Among them the vast majority (44) are solutions of Riccati equations, 2 satisfy~$q$-Riccati equations and the remaining~7 satisfy difference equations.
\end{obs}

We now give more detail on the calculations in the differential case and then outline the variants of our method in the~$q$-difference and difference cases. An implementation in the differential case is available under the form of a submodule \texttt{gfun:-ContFrac} of the package \texttt{gfun} (for versions $\ge3.70$). It can be downloaded from our web pages. All the examples of solutions of Riccati equations from~\cite{CuytPetersenVerdonkWaadelandJones2008} are provided through the associated help pages.

\subsection{C-fractions from Differential Equations}
In our experiment, the Riccati equations were themselves found by a guessing approach on power series expansions to small order (less than~$30$). Depending on how one decides to define the power series from the computer algebra point of view, these Riccati equations can also be automatically proved to hold. 

\paragraph{Gauss's continued fraction}
The classical hypergeometric series is
\[{}_2F_1(a,b;c;z):=\sum_{n\ge0}{\frac{(a)_n(b)_n}{(c)_n}\frac{z^n}{n!}},\]
where~$(a)_n$ is the Pochhammer symbol~$(a)_n=a(a+1)\dotsm(a+n-1)$.
Gauss proved the following identity
\begin{gather*}
\frac{{}_2F_1(a,b;c;z)}{{}_2F_1(a,b+1;c+1;z)}=1+\cf{m=1}{\infty}{a_mz}{1},\\
a_{2k}=-\frac{(k+b)(k+c-a)}{(2k+c)(2k-1+c)}, 
a_{2k+1}=-\frac{(k+a)(k+c-b)}{(2k+c)(2k+1+c)},
\end{gather*}
for the quotient of two contiguous hypergeometric series. This is the source of many continued fractions for special functions by specialization of the parameters.
If~$y=1+F$ is the function on the left-hand side, then elementary properties of the~${}_2F_1$ that can be derived from the first order recurrences satisfied by its coefficients show that
\[cz(z-1)y'=a(c-b)z + (c(a-b)z+c^2)y + c^2y^2.\]
This is our starting point. From there, it is easy to compute the first 20~coefficients and conjecture the formulas for~$a_{2k}$ and~$a_{2k+1}$ by rational interpolation. As in Eq.~\eqref{eq:rec-even}, a recurrence for even indices follows. From Corollary~\ref{coro}, a linear recurrence of order~4 follows for the remainder~$H_{2k}$, that can be either obtained by hand from Proposition~\ref{prop:riccati}, or by a generic code that searches for linear dependency. 
Next, reduction of order gives a two-term linear recurrence within a couple of seconds:
\[
H_{2(n-2)}={z}^{2}
\frac{\left( n+a \right)  \left( n-a+c \right)  \left( n+b \right)
 \left( n-b+c \right)}
 {\left( 2n+c \right) ^{2} \left( 2n+c-1 \right)^{2}}  H_{2(n-3)}
\]
and this concludes the automatic proof.

\paragraph{More parameters} 
Khovanskii~\cite[p.~85]{Khovanskii1963} gives an explicit continued fraction with 5 parameters for the power series solution of the differential equation 
\[ (1+\alpha z)zy'+(\beta+\gamma z)y+\delta y^2=\epsilon z,\; y(0)=0\]
(an extra parameter~$k$ is obtained by changing~$z$ into~$z^k$ and adjusting the coefficient of~$y'$; we have relabeled the parameters). This contains the equation for Gauss's continued fraction above as a special case.

From there again rational formulas for~$a_{2k}$ and~$a_{2k+1}$ are obtained by guessing on the first 20 values; a recurrence of order 4 can be found for the remainders~$H_{2k}$; Algorithm~\ref{reducerecorder} reduces it to the conclusive recurrence:
\begin{multline*}
(2n+\beta)^2 (2n+\beta-1)^2 H_{2(n-2)}=\\
  -(\alpha n^2+(\alpha\beta+\gamma) n+\beta \gamma+\delta \epsilon) (\alpha n^2+(\alpha \beta-\gamma) n+\delta \epsilon)z^2 H_{2(n-3)}.
\end{multline*}

\paragraph{Other examples} 
We also applied our method to a few functions not mentioned by Cuyt {\it et alii}~\cite{CuytPetersenVerdonkWaadelandJones2008} and in particular found (and proved) experimentally the following nice C-fraction for the Airy function:
\begin{multline*}
  z\frac{\operatorname{Ai}'}{\operatorname{Ai}}(1/z^2)=-1-\frac {z^3}4\Big/\left(1+\cf{m=2}{\infty}{a_m(z)}1\right), \\ a_{2k}=(6k-1)z^3/8,\quad a_{2k+1}=(6k+1)z^3/8.
\end{multline*}
It also follows from known C-fractions for the divergent~${}_2F_0$.

\subsection{q-analogues}
The method used in this article also applies to~$q$-analogues. We outline the very simple example of the~$q$-exponential:
\[e_q(z):=\sum_{m\ge0}{\frac{(1-q)^m}{(1-q)(1-q^2)\dotsm(1-q^m)}z^m},\]
which satisfies the $q$-differential equation
\begin{equation}\label{eq:qdeq-exp}
\frac{e_q(qz)-e_q(z)}{(q-1)z}-e_q(z)=0.
\end{equation}
The classical exponential is obtained as the limit when~$q\rightarrow1$.
The first coefficients of the continued fraction expansion let one guess~$a_1=z$,
\[a_{2k}=-\frac{q^{k-1}(1-q)z}{(1+q^k)(1-q^{2k-1})},\quad a_{2k+1}=\frac{q^{2k}(1-q)z}{(1+q^k)(1-q^{2k+1})},\]
a clear generalization of the continued fraction~\eqref{eq:cf-exp} for~$\exp$. In order to prove this continued fraction, the recurrence for~$(P_{2k})$ is computed as in Section~\ref{sec:periods}, which gives
\[
P_{2k+2}=\left(1-\frac{(1-q)q^kz}{(1+q^k)(1+q^{k+1})}\right)P_{2k}
+\frac{q^{3k-1}(1-q)^2z^2}{(1+q^k)^2(1-q^{2k+1})(1-q^{2k-1})}P_{2k-2}.
\]
The sequence~$H_{k}$ is defined as the numerator of the evaluation of~\eqref{eq:qdeq-exp}, namely
\[\frac{P_{k}(qz)Q_k(z)-P_k(z)Q_k(qz)}{(q-1)z}-P_k(z)Q_k(qz)-Q_k(z)Q_k(qz).\]
Next, we compute a linear dependency between~$H_{2k},H_{2k+2},\dots$, which is still of order~4 (but significantly bigger than its limit as~$q\rightarrow1$).
The~$q$-analogue of reduction of order then proves
\[H_{2k+2}=-\frac{q^{3k+2}(1-q)^2z^2}{(1+q^{k+1})^2(1-q^{2k+1})^2}H_{2k},\]
which concludes the proof, providing with a generalization of the expression for exp that is recovered by letting~$q\rightarrow1$.

Using the same steps leads to an automatic proof of Heine's~$q$-analogue of Gauss's continued fraction~\cite[19.2.1]{CuytPetersenVerdonkWaadelandJones2008}: the~$q$-hypergeometric series is defined by
\[{}_2\phi_1(a,b;c;q;z)=\sum_{n\ge0}{\frac{(a;q)_n(b;q)_n}{(c;q)_n}\frac{z^n}{(q;q)_n}},\]
where~$(a;q)_n$ is the~$q$-Pochhammer symbol \[(a;q)_n=(1-a)(1-aq)\dotsm(1-aq^{n-1}).\] Heine's continued fraction is
\begin{gather*}
\frac{{}_2\phi_1(a,b;c;q;z)}{{}_2\phi_1(a,bq;cq;q;z)}=1+\cf{m=1}{\infty}{a_mz}{1},\\
a_{2k+1}=\frac{(1-aq^k)(cq^k-b)q^k}{(1-cq^{2k})(1-cq^{2k+1})},\qquad
a_{2k}=\frac{(1-bq^k)(cq^k-a)q^{k-1}}{(1-cq^{2k-1})(1-cq^{2k})}.
\end{gather*}
We sketch the main steps of the computation. The~$q$-Riccati equation is
\begin{multline*}
(1-c)^2F(z)F(qz)+(1-c)(bz-c)F(qz)
\\+(1-z)(1-az)F(z)-z(a-1)(b-c)=0.
\end{multline*}
The sequence~$H_k$ of interest is therefore the numerator of the evaluation of this left-hand size at~$F(z)=P_{k}(z)/Q_{k}(z)$. The continued fraction being periodic of period~2, a recurrence for~$H_{2k}$ (or order 4) is computed. Reduction of order yields
\[\frac{H_{2k+2}}{z^2H_{2k}}=\frac{(1-aq^{k+1})(1-bq^{k+1})(a-cq^{k+1})(b-cq^{k+1})q^{2k+1}}
{(1-cq^{2k+1})(1-cq^{2k+2})},\]
which concludes the proof. This automates 2~more of the formulas in~\cite{CuytPetersenVerdonkWaadelandJones2008}.

\subsection{Difference Equations} 
The same method applies to difference equations. For instance, it results in one of the classical proofs~\cite[chap.~3]{Khrushchev2008} of Brouncker's continued fraction for
\[b(s):=\left(\frac{\Gamma\left(\frac{s+1}{4}\right)}{\Gamma\left(\frac{s+3}{4}\right)}\right)^2,\]
where~$\Gamma$ is Euler's Gamma function. Using the functional equation~$\Gamma(s+1)=s\Gamma(s)$, it follows that~$b(s)$ satisfies
\[b(s)b(s+2)={16}/{(s+1)^2}.\]
Looking for a formal power series solution in inverse powers of~$s$ (and nonnegative leading term) leads to a unique solution
$b(s)=4/s-2/{s^3}+\dotsb$
This is then converted into a continued fraction expansion with coefficients~$(a_k)$ given by
\[\frac{4}{s},\frac{1}{2s^2},\frac{9}{4s^2},\frac{25}{4s^2},\frac{49}{4s^2},\frac{81}{4s^2},\dotsc\]
from which it is easy to conjecture~$a_k=(2k-3)^2/(4s^2)$ for~$k\ge3$. The analogue of~$H_k$ in this context is
\[H_k=(s+1)^2P_k(s)P_k(s+2)-16Q_k(s)Q_k(s+2),\]
for which the same approach as above produces a linear recurrence of order~4
which is not sufficient to conclude that the valuations increase. From there, reduction of order with Algorithm~\ref{reducerecorder}
yields the shorter
\[H_{k+1}=-\frac{(2k+1)^2}{4s(s+2)}H_k,\quad k\ge1,\]
which exhibits the required increase of valuations.

The same technique has been applied to all the explicit C-fractions concerning the~$\psi$ function in~\cite{CuytPetersenVerdonkWaadelandJones2008}, thereby completing the experiment on this book.

\section{Conclusion}

In a simple and unified way, our approach to continued fractions recovers an  unexpectedly large number of explicit C-fractions from the literature. One miracle that takes place is that in all cases, the sequence of remainder polynomials turns out to be hypergeometric or $q$-hypergeometric. We are currently exploring this phenomenon in more detail.

\section*{Acknowledgements} This work was supported in part by the project FastRelax ANR-14-CE25-0018-01.

\bibliographystyle{abbrv}

\end{document}